\newif\ifJournal
\Journalfalse
\ifJournal
\RequirePackage{fix-cm}
\documentclass[twocolumn]{svjour3} % Use option that matches published journal style
\smartqed
\usepackage{flushend}
\else
\documentclass[11pt]{article}
\usepackage[margin=1.25in]{geometry}
\usepackage{mathptm}
\usepackage{amsthm}
\newtheorem{theorem}{Theorem}[section]

\usepackage{hyperref}
\fi

\usepackage{graphicx}
\usepackage{cite}

\newcommand{\refine}{\mathop{\mathrm{refine}}}
\newcommand{\oversized}{\mathop{\mathrm{oversized}}}
\newcommand{\distance}{\mathop{\mathrm{distance}}}

\let\oldendproof\endproof
\def\endproof{\qed\oldendproof}

\begin{document}

\title{Diamond-Kite Adaptive Quadrilateral Meshing\thanks{This research was supported in part by the National Science Foundation under grants 0830403 and 1217322, and by the Office of Naval Research under MURI grant N00014-08-1-1015. Some of these results appeared in preliminary form in a paper in the 21st International Meshing Roundtable, from which this paper is adapted. We thank Scott Mitchell and Damrong Guoy for helpful discussions.}}

\ifJournal
\author{David Eppstein}
% Use \authorrunning{Short Title} for an abbreviated version of
% your contribution title if the original one is too long
\institute{Department of Computer Science, Univ. of California, Irvine\\
\email{eppstein@ics.uci.edu}}

\date{Received: date / Accepted: date}
% The correct dates will be entered by the editor

\else
\author{David Eppstein\\
Department of Computer Science, Univ. of California, Irvine}
\date{ }

\maketitle

\begin{abstract}
We describe a family of quadrilateral meshes based on \emph{diamonds}, rhombi with $60^\circ$ and $120^\circ$ angles, and \emph{kites} with $60^\circ$, $90^\circ$, and $120^\circ$ angles, that can be adapted to a local size function by local subdivision operations. Our meshes use a number of elements that is within a constant factor of the minimum possible for any mesh of bounded aspect ratio elements, graded by the same local size function, and is invariant under Laplacian smoothing.
The vertices of our meshes form the centers of the circles in a pair of dual circle packings. The same vertex placement algorithm but a different mesh topology gives a pair of dual well-centered meshes adapted to the given size function.
\end{abstract}

\ifJournal
\fi

\section{Introduction}

In unstructured mesh generation for finite element simulations, it is important for
a mesh to have elements that are well-shaped (as this controls the convergence of the method), small enough to fit the features of the problem domain and its solution, and yet not so small that there are an unnecessarily large number of elements, slowing down the time per iteration of the simulations. 
In this paper we describe a recursive subdivision algorithm for generating quadrilateral meshes that are guaranteed to have all of these properties: they have bounded aspect ratio elements, the elements are sized to match to a given local size function, and the meshes are guaranteed to use a number of elements that is within a constant factor of the minimum number required by that size function.  The elements of our meshes have two shapes:  \emph{diamonds}, rhombi with $60^\circ$ and $120^\circ$ angles, and \emph{kites} with $60^\circ$, $90^\circ$, and $120^\circ$ angles; therefore, we call our meshes \emph{diamond-kite meshes}. They may be generated in time linear in the number of mesh elements, and adapted by local refinement and coarsening when the size function changes.

As well as having bounded aspect ratio elements and having an approximately-minimum number of elements, diamond-kite meshes obey a number of other interesting properties:
\begin{itemize}
\item It is possible to place circles centered at each vertex of a diamond-kite mesh in such a way that pairs of circles cross at right angles if they correspond to adjacent vertices in the mesh, are tangent to each other if they correspond to nonadjacent vertices of a common quadrilateral, and otherwise are disjoint from each other. This system of circles forms a \emph{circle packing} of a type studied by Thurston~\cite{Thu-MSRI-02}; circle packings with more irregular structures have previously been used in meshing, but this is the first known mesh that corresponds to this type of highly structured circle packing.
\item The set of diamond-kite meshes has the structure of a distributive lattice, in which every pair of meshes has a coarsest common refinement and a finest common coarsening.
\item The faces of a diamond-kite mesh may be colored with three colors, a property that is useful in scheduling parallel updates to the values stored at mesh elements~\cite{BenFlaKri-AMHCS-92,BerEppHut-Algo-02} and that is not true of all quadrilateral meshes.
\item Diamond-kite meshes are invariant under Laplacian smoothing, a commonly used technique for mesh improvement in which each vertex of a mesh is moved to the centroid of its neighbors. Thus, every diamond-kite mesh forms a Tutte embedding~\cite{Tut-PLMS-63} of its underlying graph.
\item The vertices of a diamond-kite mesh can be used to form two dual \emph{well-centered meshes} with polygonal elements that have up to six sides. In these meshes, each mesh vertex lies in the interior of its dual face, and the primal and dual edges cross each other at right angles, properties that are important in certain numerical methods~\cite{VanHirGuo-SJSC-09}. In previous methods for well-centered meshing, one of the two dual meshes was a triangulation and the other a Voronoi diagram, but the two dual meshes generated from a diamond-kite mesh are both of the same type as each other.
\end{itemize}

\subsection{Related work}

As a recursive subdivision scheme based on a regular tiling of the plane, our meshing method bears a strong resemblance to meshing based on \emph{quadtrees}, a recursive subdivision of squares into smaller squares.
The first triangular mesh generation algorithms to provide theoretical guarantees on both the element aspect ratio and the total number of elements for a given local size function used quadtrees~\cite{BerEppGil-JCSS-94} and they had long been used in meshing prior to its theoretical justification~\cite{YerShe-CGA-83,FreMar-IMR-98}. However, in contrast to the subdivision we use, quadtrees are not usually used as meshes directly, because of the potentially large number of neighbors of each cell; instead,  quadtree squares are typically further subdivided into mesh elements. For instance, in a \emph{balanced} quadtree (one in which neighboring squares differ in size by at most a factor of two) it is possible to partition each quadtree square into a constant number of triangular elements that all have the shape of an isosceles right triangle, giving a triangle mesh with bounded aspect ratio elements.

Other recursive subdivision schemes have also previously been used in meshing, including the recursive subdivision of triangles into four smaller triangles used by the Sloan digital sky survey~\cite{SzaKunTha-SIGMOD-00}, and the hexagon-based meshing algorithms of Su{\ss}ner, Greiner, Liang, and Zhang~\cite{SusGre-IMR-09,LiaZha-CMAME-11}. These methods, like ours, are based on a triangular or hexagonal lattice, but they scale this lattice by factors of two at each level of subdivision whereas we use factors of $\sqrt 3$. Other subdivision schemes related to the methods described here include honeycomb refinement~\cite{DynLevLiu-CAD-92,Zor-SIGGRAPH-99} and $\sqrt 3$ refinement~\cite{Kob-SIGGRAPH-00,Zor-SIGGRAPH-99}; however, these schemes are typically applied to every element of a structured or semistructured mesh rather than (as here) to selected elements of an unstructured mesh.

Our method is also closely related to the use of circle packings in meshing. Several unstructured mesh generation algorithms work by packing circles into the domain to be meshed, and then using the positions of the circles to guide the generation of a final mesh. The circle packing phases of these algorithms have been based on multiple ideas, including placing circles one at a time using a greedy algorithm~\cite{BerMitRup-DCG-95,BerEpp-IMR-97,Epp-IJCGA-97,LiTenUng-IJNME-00,LiTenUng-IMR-99,WanMinLo-AES-07,LiuLiChen-AMS-08},  physical simulation~\cite{ShiGos-SMA-95}, and the decimation of quadtree-based oversampling schemes~\cite{MilTalTen-STOC-95,MilTalTen-SJNA-99}.  These algorithms have been used to find nonobtuse triangular meshes for polygonal domains~\cite{BerMitRup-DCG-95,Epp-IJCGA-97} as well as bounded-aspect-ratio triangular meshes~\cite{LiTenUng-IJNME-00} and quadrilateral meshes in which all elements belong to certain special types of quadrilateral~\cite{BerEpp-IMR-97}. The circle packings generated by these methods can be made to have radii controlled by a local size function~\cite{LiTenUng-IJNME-00,WanMinLo-AES-07}, and mesh generation techniques based on these methods can be applied in higher dimensions as well~\cite{LiTenUng-IMR-99,LoWan-CMAME-05,MilTalTen-STOC-95,MilTalTen-SJNA-99}.
However, because the circle packings  generated by these methods are irregular in structure, these methods do not use the circle centers as the only vertices of their meshes. Instead, they add additional vertices at points such as the circumcenters of the gaps between circles, and they typically also require additional case analysis to handle the different possible shapes of their gaps. In contrast, we generate a circle packing from a mesh rather than generating a mesh from a circle packing, but in so doing we obtain a more direct relationship between the mesh and the packing: the mesh vertices are exactly the centers of the packed circles.

\section{Diamond-kite meshes}

Our construction begins with the \emph{rhombille tiling} of Figure~\ref{fig:rhombille} (left), a tessellation of the plane by rhombi with $60^\circ$ and $120^\circ$ angles, formed by subdividing a hexagonal tiling into three rhombi per hexagon~\cite{ConBurGoo-TSoT-08}. Each vertex of the rhombille tiling has degree (valency) either three or six: either six rhombi meet at their acute corners, or three rhombi meet at their obtuse corners. The short diagonals of the rhombi form another hexagonal tiling in which each hexagon surrounds a degree-six vertex of the rhombille tiling (Figure~\ref{fig:rhombille}, right).

\begin{figure}[t]
\centering\includegraphics[width=1.5in]{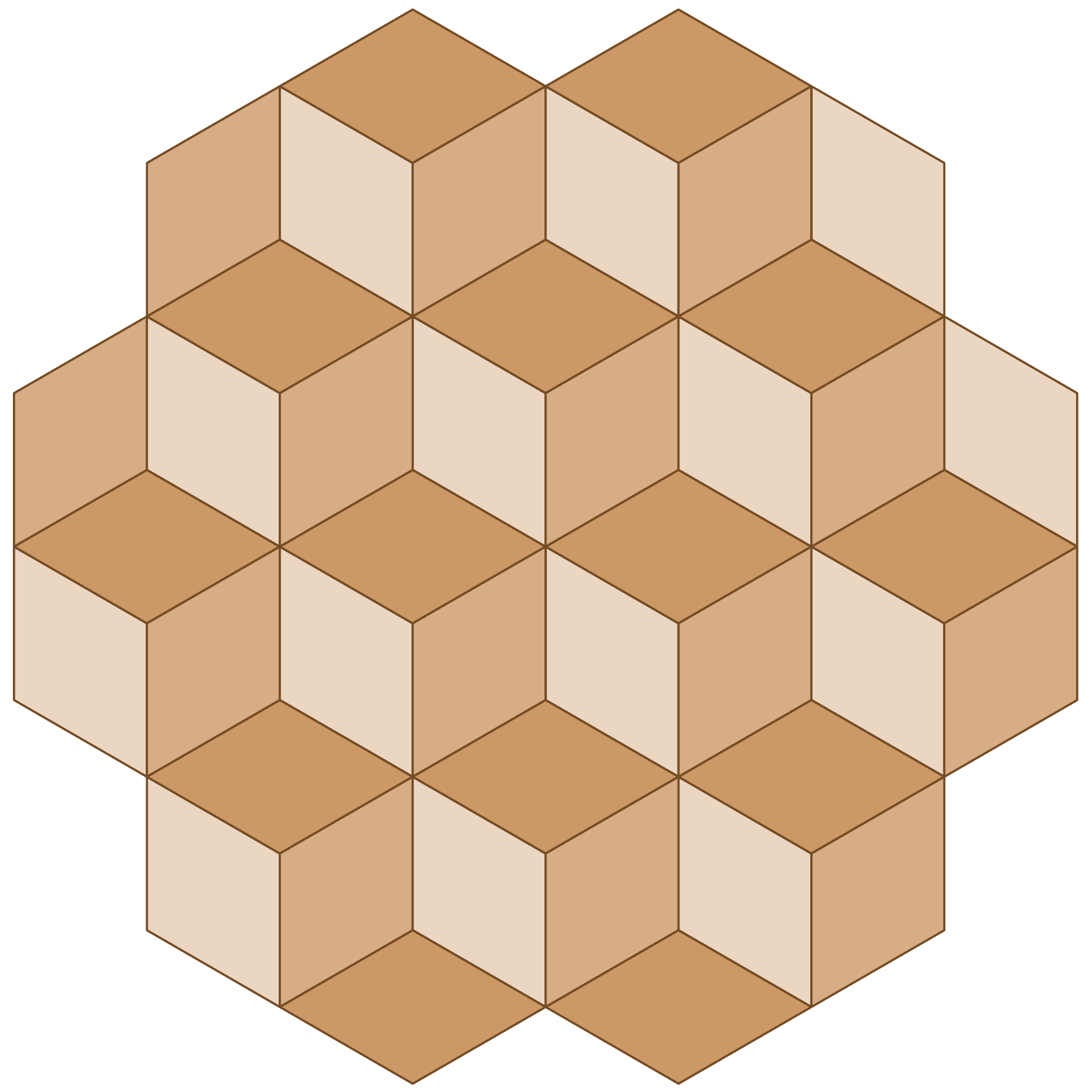}
\qquad\includegraphics[width=1.5in]{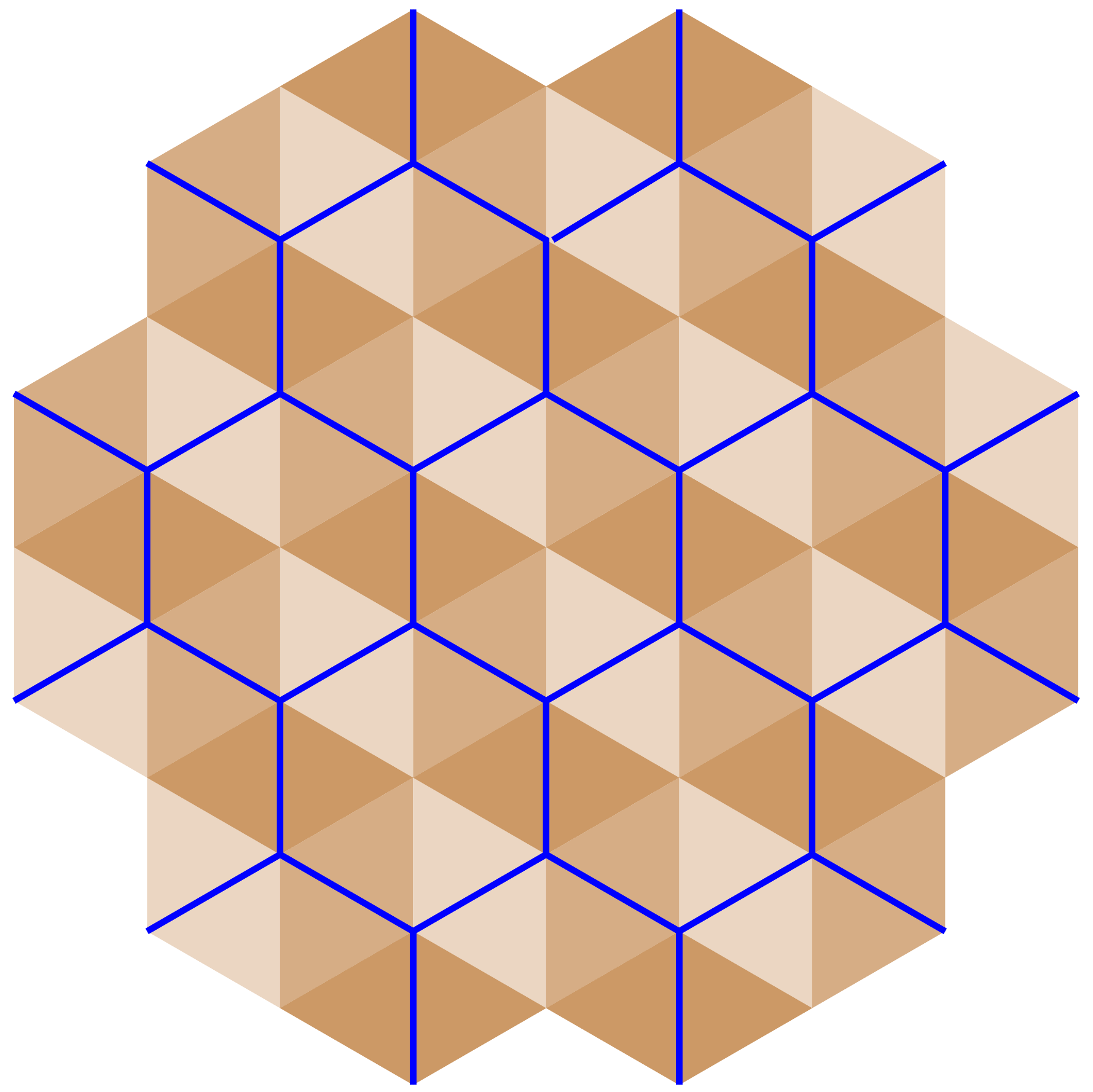}
\caption{Left: The rhombille tiling. Right: The short diagonals of the rhombille tiling form a hexagonal tiling.}
\label{fig:rhombille}
\end{figure}

\begin{figure}[t]
\centering\includegraphics[height=1.25in]{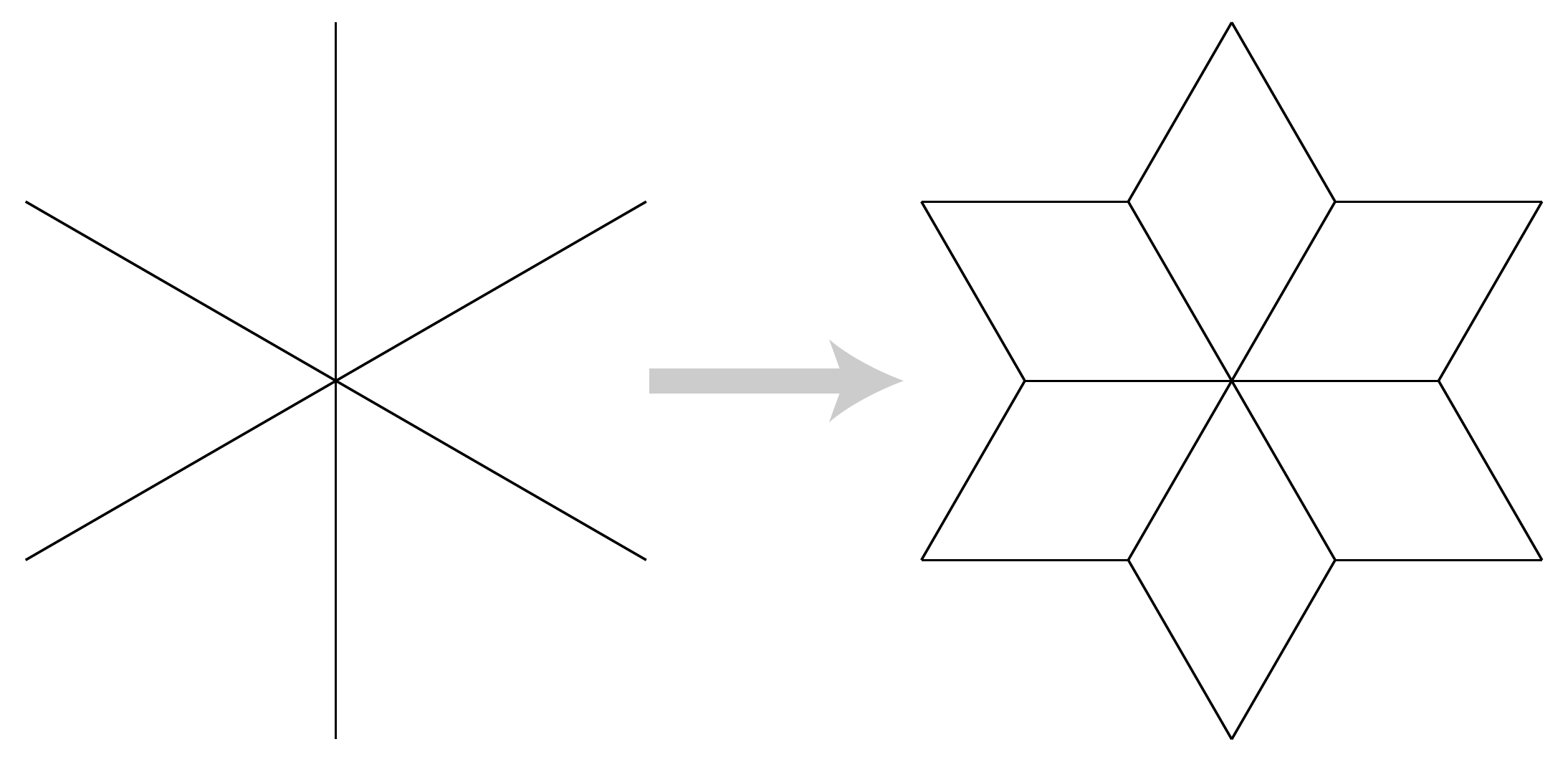}
\caption{Replacement of six edges within a hexagon (left) by six rhombi with sides shorter by a factor of $1/\sqrt 3$ than the replaced edges (right).}
\label{fig:local-replacement}
\end{figure}

Suppose that we wish to refine a rhombille tiling within a region $R$ of the plane, forming a mesh with smaller elements, while leaving the tiling unchanged outside $R$. To do so, we may approximate $R$ by a set of the hexagons formed by short diagonals, and perform the local replacement operation illustrated in Figure~\ref{fig:local-replacement} within each hexagon. This operation replaces the six edges that meet in the center of the hexagon with a network of 18 edges, shorter by a factor of $1/\sqrt 3$ from the original edges. These new edges form the boundaries of six rhombi similar to the ones in the original rhombille tiling but rotated from them by $30^\circ$ angles. Each subdivided quadrilateral crossing the boundary of the hexagon remains a quadrilateral after the replacement, so the result of the replacement is a valid quadrilateral mesh with six additional quadrilaterals for every replaced hexagon.

Figure~\ref{fig:subdivided-rhombille} shows the mesh resulting from multiple  local replacements in a rhombille tiling. When one of the replaced hexagons shares an edge with a hexagon that has not been replaced, the quadrilaterals that lie across the shared boundary of the two hexagons are kite-shaped, with vertex angles $60^\circ$, $90^\circ$, and $120^\circ$. However, when two or more replaced hexagons share an edge with each other, the quadrilaterals that lie across their shared boundary are rhombi congruent to the ones contained within each replaced hexagon. Within any region formed by multiple replaced hexagons, the smaller rhombi formed by this replacement process meet up with each other in the pattern of another rhombille tiling, rotated from the original tiling by a $30^\circ$ angle and with tiles that are smaller by a $1/\sqrt 3$ factor.

\begin{figure}[t]
\centering\includegraphics[width=2.5in]{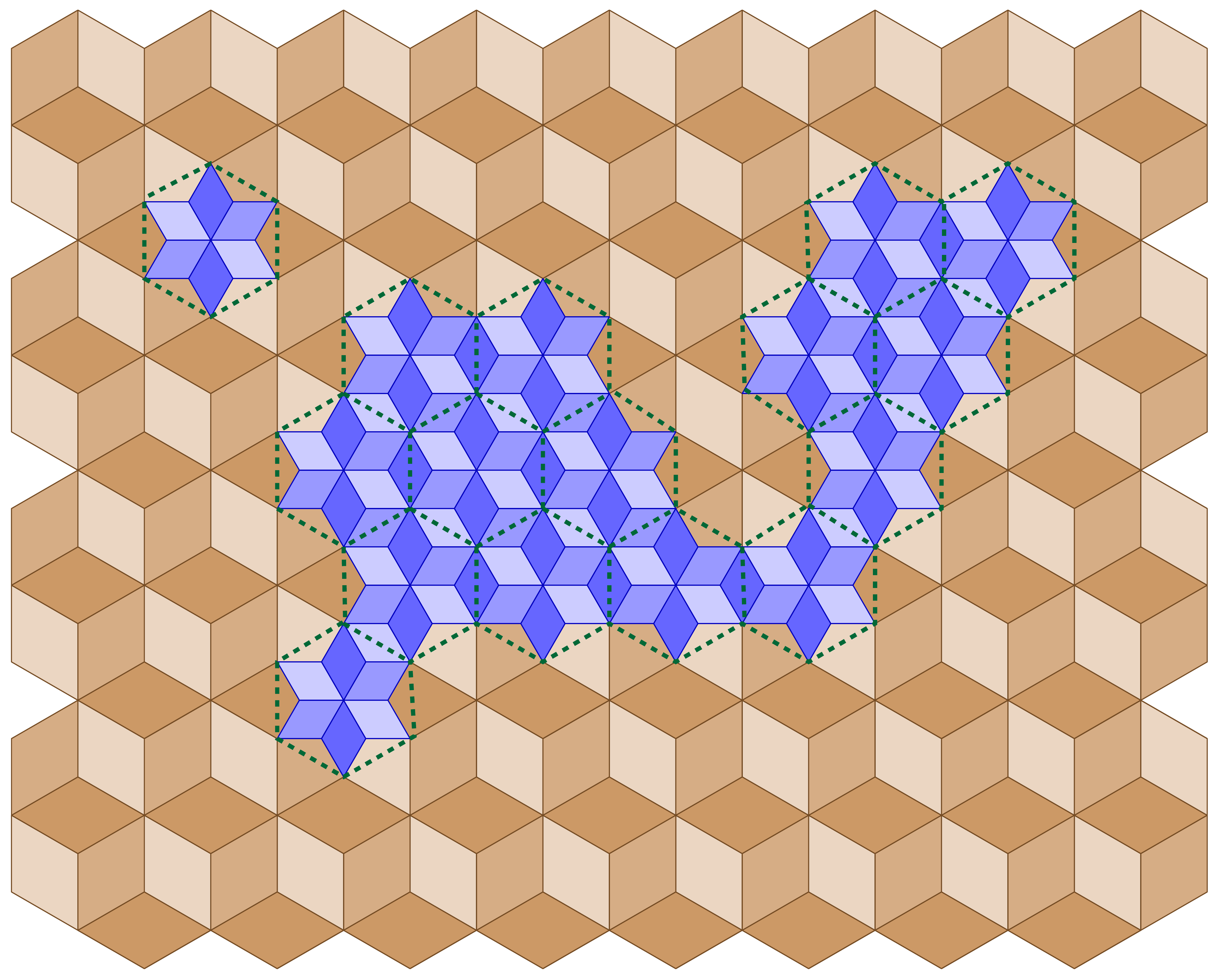}
\caption{The result of performing multiple local replacements in a rhombille tiling. Within the blue replaced region, we have another rhombille tiling, rotated from the original and with smaller tiles.}
\label{fig:subdivided-rhombille}
\end{figure}

Once this replacement has been performed, the same replacement process can be performed within the finer rhombille tiling formed within the replaced region. The degree-six vertices that can be replaced within this finer tiling are either the same degree-six vertices that were replaced previously, or the points where three replaced hexagons meet.

We call the meshes generated by any number of steps of this replacement process ``diamond-kite meshes''.

\section{Adaptation to a local size function}
\subsection{Prerequisite structure of replacement operations}

Given an initial rhombille tiling $T$, we may uniquely describe each of the local replacement steps used to form a diamond-kite mesh by a pair of parameters $(p,s)$, where $p$ is the center point of the hexagon within which the replacement happens and $s$ is its side length. A replacement step with parameters $(p,s)$ may be performed at most once within a sequence of replacements to a mesh, and it may only be performed when the mesh that results from the earlier replacements in the sequence has exactly six edges of length $s$ meeting at point $p$.

This condition that the mesh is ready for replacement $(p,s)$ to be performed can also be stated indirectly in terms of certain \emph{prerequisite} replacement steps that must be performed prior to $(p,s)$, instead of describing the configuration surrounding~$p$.  To determine these prerequisites, we may analyze cases to determine the conditions under which a replacement with parameters $(p,s)$ is possible:
\begin{itemize}
\item If $s$ is the length of the sides of the original tiles of $T$, then replacement $(p,s)$ may always be performed in every diamond-kite mesh formed from $T$ in which it has not already been performed.
\item If $s$ is smaller than the side length in $T$, and $(p,s\sqrt 3)$ is the pair of parameters for another replacement step, then replacement $(p,s)$ may be performed if and only if replacement $(p,s\sqrt 3)$ has already been performed. The reason for this is that the replacement $(p,s\sqrt 3)$ is the only possible way to incorporate into the tiling the six edges that will be replaced by $(p,s)$. Again, replacement $(p,s)$ may always be performed in every mesh in which it has not already been performed and in which $(p,s\sqrt 3)$ has been performed.
\item In the remaining case, there are three points $p_0$, $p_1$, and $p_2$ equally spaced around $p$ at distance $s\sqrt 3$ from it, such that each point $p_i$ gives the parameterization of a replacement step $(p_i,s\sqrt 3)$ and such that replacement $(p,s)$ may be performed if and only if all three of $(p_i,s\sqrt 3)$ have already been performed. Each of the three replacements $(p_i,s\sqrt 3)$ is the only possible way of incorporating into the tiling two of the six edges that will be replaced by $(p,s)$. As before, once it becomes possible to perform $(p,s)$, it remains possible to perform it until it actually is performed.
\end{itemize}

This prerequisite structure may be described by an infinite directed acyclic graph $G_T$ which has a vertex for each pair $(p,s)$ that defines a valid replacement step, and an edge from $(p,s)$ to each of the prerequisite replacement operations described in the case analysis above. That is, for every $(p,s)$, if $(p,s\sqrt 3)$ defines a valid replacement step then graph $G_T$ has an edge from $(p,s)$ to $(p,s\sqrt 3)$. In the case where there are three prerequisites $(p_i,s\sqrt 3)$, graph $G_T$ will instead have three edges going out of $(p,s)$, one to each of these three prerequisites.

This graph may be used to define an infinite partially ordered set $P_T$ that has an element for each pair $(p,s)$, and in which two elements $x$ and $y$ are ordered $x\le y$ whenever there is a directed path from $x$ to $y$ in $G_T$. In this partially ordered set, the set of all replacement steps that must be performed prior to step $(p,s)$ is the set $\{(p',s')\mid (p',s')<(p,s)\}$; it includes not just the immediate prerequisites to $(p,s)$, but also the prerequisites of the prerequisites, and so on.

\subsection{The lattice of mesh refinements}
If $T'$ is a diamond-kite mesh formed by refining the initial mesh $T$, then (by the analysis above) the set of replacement steps that were used to generate $T'$ from $T$ must be a finite \emph{lower set} in partial order $P_T$, that is, a set $L$ of elements with the property that, for every element $y\in L$ and for every element $x$ with $x\le y$, $x$ is also in $L$. Conversely, every finite lower set $L$ uniquely defines a diamond-kite mesh $T_L$ as a refinement of $T$, because the replacement operations in $L$ may be performed in any order that is consistent with the case analysis above; different orderings of the same operations will always lead to the same results. The orderings in which a given lower set $L$ of replacement steps may be performed are exactly the linear extensions of the partial order induced by the elements of $L$ in $P_T$, and a valid ordering for a given lower set $L$ may be calculated by applying a topological sorting algorithm to the directed acyclic graph induced as a subgraph of $G_T$ by the elements of $L$.

\begin{theorem}
The set of diamond-kite meshes formed from a given initial mesh has the structure of a distributive lattice, in which every two meshes $T_1$ and $T_2$ have a unique  finest common coarsening $T_1\wedge T_2$ and a unique coarsest common refinement $T_1\vee T_2$.
\end{theorem}

\begin{proof}
Let $T_1$ and $T_2$ be any two diamond-kite meshes formed by refining $T$, and described by the respective finite lower sets $L_1$ and $L_2$ in $P_T$. Then the sets $L_1\cap L_2$ and $L_1\cup L_2$ are also finite lower sets, describing respectively the finest mesh $T_1\wedge T_2$ from which both $T_1$ and $T_2$ can be formed by refinement, and the coarsest mesh $T_1\vee T_2$ that can be formed by refining both $T_1$ and $T_2$. In this way, as with the lower sets of every partially ordered set~\cite{Bir-DMJ-37}, the family of diamond-kite meshes can be given the structure of a distributive lattice.
\end{proof}

We remark that the same ideas of constructing an infinite graph describing the prerequisite relation between potential replacement steps, deriving an infinite partial order from the graph, and describing each possible mesh as a lower set of this partial order, can be applied equally well to describe the set of possible balanced quadtrees derived from an initial square. Thus, the set of balanced quadtrees can also be given the structure of a distributive lattice.

\subsection{Local replacement at mesh vertices}

\begin{figure}[t]
\centering\includegraphics[width=3in]{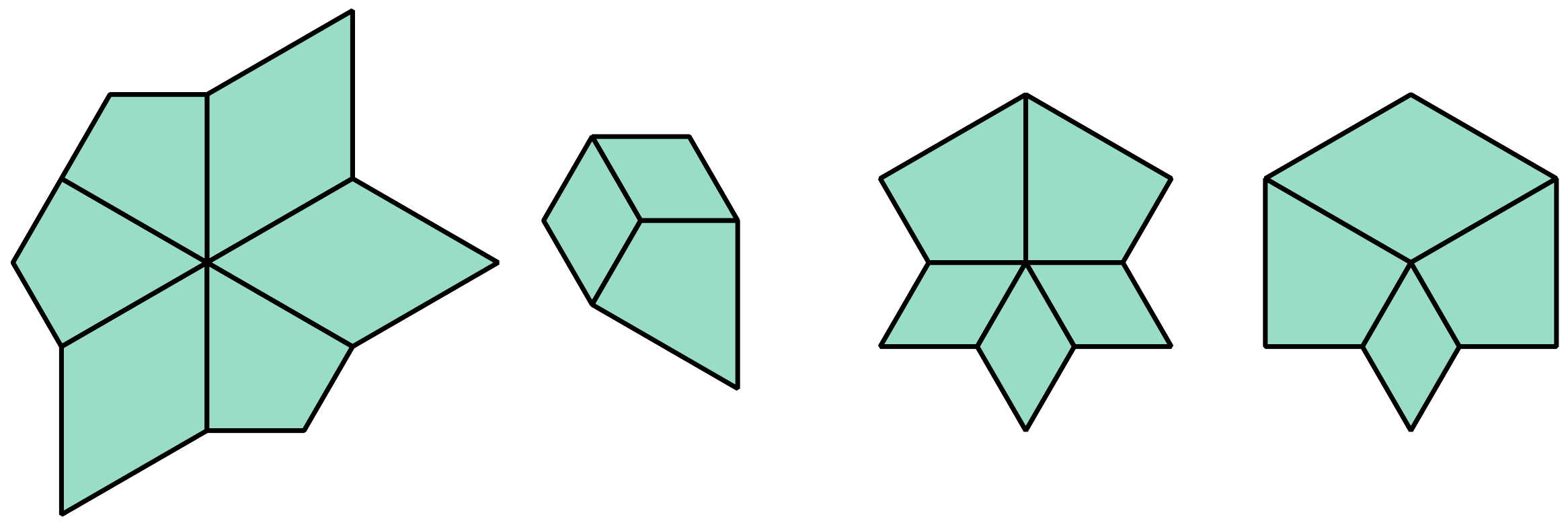}
\caption{Cases for $\refine(p)$. From left to right: (a) $p$ has degree six, and is ready for immediate replacement; (b) $p$ has degree three, and does not meet the preconditions of $\refine$; (c) $p$ has degree five, and a replacement at the $60^\circ$ kite vertex must be performed prior to a replacement at $p$; and (d) $p$ has degree four, and two $60^\circ$ kite vertices must be replaced prior to a replacement at~$p$.}
\label{fig:refine-cases}
\end{figure}

We now define a recursive subdivision algorithm, which we refer to by the subroutine name $\refine(p)$, that performs a single local replacement step at a vertex $p$ of a diamond-kite mesh, after performing all prerequisite replacements. We require, as a precondition for this algorithm, that $p$ be a $60^\circ$ vertex of at least one mesh element; the replacement performed by this subroutine will be the one parametrized by $(p,s)$, where $s$ is the length of the mesh edges on either side of the $60^\circ$ angle.

More specifically, $\refine(p)$ performs the following two steps:
\begin{enumerate}
\item For each kite that has a $90^\circ$ angle at $p$, let $q$ be the $60^\circ$ angle angle of the kite, and call $\refine(q)$ recursively.
\item Perform a replacement step at $p$.
\end{enumerate}

\begin{theorem}
A call to $\refine(p)$ performs the replacement step $(p,s)$ (where $s$ is as defined above) and all its prerequisites, but no other replacements. The time for the procedure is proportional to the number of new elements added by these replacement steps.
\end{theorem}

\begin{proof}
To verify that these steps perform the prerequisites of replacement step $(p,s)$, and step $(p,s)$ itself, but do not perform any other replacement steps, we go through a case analysis that examines the possible local configurations near $p$.
\begin{itemize}
\item If $p$ has degree six, then the most recent replacement step that affected the edges incident to $p$ must either have replaced a hexagon with $p$ at its center, or have been the third of three replacements of hexagons meeting at~$p$. In this case, $p$ is already surrounded by diamonds and/or kites having $60^\circ$ angles at $p$, as shown in Figure~\ref{fig:refine-cases}(a). In this case, it is possible to perform a local replacement step at $p$ without performing any other replacements. Since there are no $90^\circ$ angles at $p$, the first step of $\refine$ makes no local calls, and $\refine$ correctly makes only the replacement step $(p,s)$ itself.
\item If $p$ has degree three, then it must be the case that the most recent replacement step at $p$ replaced a hexagon for which $p$ was interior but not central. Then $p$ is a $120^\circ$ vertex of three elements, either two diamonds and one larger kite (as in  Figure~\ref{fig:refine-cases}(b)) or three diamonds. In this case, it is not possible for the precondition of the $\refine$ algorithm to be met, because there is no $60^\circ$ angle at $p$. Thus, in this case it does not matter what might happen when $\refine(p)$ is called.
\item If $p$ has degree five, then it must be the case that the most recent replacement to affect the neighborhood of $p$ was the replacement of a hexagon having $p$ as a vertex, and additionally this must have been the second replacement of the three hexagons of that size meeting at $p$. Then the neighborhood of $p$ consists of three elements with $60^\circ$ angles (diamonds or kites within the two replaced hexagons) and two elements with $90^\circ$ angles (necessarily kites in the third hexagon); see  Figure~\ref{fig:refine-cases}(c). Let $q$ be the shared $60^\circ$ vertex of these two kites. When $\refine(p)$ is called, it will recursively call $\refine(q)$, performing the prerequisite replacement step, after which we may perform a replacement step at $p$.
\item If $p$ has degree four, then $p$ was a vertex of the hexagon for the most recent replacement at $p$, and this was the first replacement of the three hexagons of that size meeting at $p$. In this case $p$ has a neighorhood with one $60^\circ$ angle (a diamond or kite in the replaced hexagon), two $90^\circ$ angles (kites in the two unreplaced hexagons), and one $120^\circ$ angle (a rhomb or kite overlapping the two unreplaced hexagons), as is depicted in Figure~\ref{fig:refine-cases}(d). The call to $\refine(p)$ will recursively call $\refine$ for each of the two $60^\circ$ vertices of the kites with incident $90^\circ$ angles, which will perform the two prerequisite replacement steps to $(p,s)$.
\end{itemize}

The recursion in step~1 of~$\refine$ necessarily terminates, because each recursive call leads to a replacement step on a larger hexagon.  Each recursive call adds elements to the mesh for the replacement step it performs, so the total time for this recursive procedure is linear in the total change to the number of elements in the mesh. The case analysis above shows that this recursion performs exactly the replacement steps that are predecessors of $(p,s)$ in the partial order $P_S$ and that had not already been performed at the start of the recursion.
\end{proof}

\subsection{Local size functions}
\label{sec:refinement}
We define a \emph{local size function} to be a function $\sigma$ that maps each point $p$ of the plane (or of a subset of the plane to be meshed) to a positive real number $\sigma(p)$, specifying the largest allowable side length of a mesh element containing~$p$.\footnote{It would be equivalent to within constant factors to specify the maximum allowable area, perimeter, diameter, or circumradius of the element, but side length turns out to be more convenient for our purposes, because it leads to fewer ambiguities about which replacement steps are necessary.} We assume that access to $\sigma$ is via a subroutine $\oversized(Q)$ that takes as argument a quadrilateral $Q$ and returns a Boolean value, true if $Q$ contains a point $p$ for which $\sigma(p)$ is less than the side length of $Q$ and false otherwise. Our task is to find a mesh that is as coarse as possible subject to the constraint that $\oversized$ returns false for all mesh elements.

To do so, we perform the following steps:
\begin{enumerate}
\item Construct an initial coarse mesh.
\item Initialize a queue $Q$ of unprocessed quadrilaterals, containing all quadrilaterals in the initial mesh.
\item Repeat until $Q$ is empty:
\begin{enumerate}
\item Find and remove a quadrilateral $q$ from $Q$.
\item If $q$ is a kite and $\oversized(q)$ returns true, let $p$ be the $60^\circ$ vertex of $q$ and call $\refine(p)$.
\item If $q$ is a diamond, let $q_1$ and $q_2$ be the two kites contained within $q$ that have the same maximum side length as $q$, and let $p_1$ and $p_2$ be their two  $60^\circ$ vertices. For each $i$ in the set $\{1,2\}$, if $\oversized(q_i)$ returns true, call $\refine(p_i)$.
\end{enumerate}
\end{enumerate}

\begin{theorem}
\label{thm:refinement}
The mesh refinement procedure described above finds the coarsest mesh consistent with the given size function, and takes time proportional to the size of that mesh.
\end{theorem}

\begin{proof}
Note that, in the diamond case,
$$\oversized(q)=\oversized(q_1)\vee\oversized(q_2),$$
because the two kites have the whole diamond as their union. Therefore, if $\oversized(q)$ is true then one or both kites will also be oversized, and processing $q$ will cause it to become subdivided. It follows that the algorithm can only terminate when there are no more oversized quadrilaterals in the refined mesh that it has produced.

Whenever this adaption procedure makes a call to $\refine(p)$, leading (after some recursive calls) to a replacement step at vertex $p$, the same replacement step must be performed in every diamond-kite mesh that is a refinement of the same initial mesh and obeys the size function $\sigma$, because otherwise the kite associated with $p$ would remain in the mesh and would have too large a side length. Therefore, the result of the refinement procedure described above is the coarsest diamond-kite mesh consistent with the size function. Each step either removes a quadrilateral from the queue without adding any others, or it takes time linear in the number of elements added, so the total time for this adaption procedure is linear in the size of the final mesh it produces.
\end{proof}

\subsection{Dynamic adaptation}

For a simulation in which the size function changes over time, it may be desirable to adapt the mesh to the new size function after each time step. To do so, we may first refine the mesh as described above (using the mesh from the previous time step as the initial mesh in the refinement algorithm), adding new smaller elements in places where the size function has diminished, and then coarsen it using a similar algorithm, removing elements in places where the size function has increased.

Coarsening (the reverse of refinement) may be performed at any vertex $v$ of the mesh that is surrounded by six diamonds, such that the none of the six kites surrounding $v$ of the next larger size than the diamonds is oversized. In this case we say that $v$ is \emph{coarsenable}; the coarsening step at $v$ consists of removing the edges of the surrounding six diamonds and replacing them by the long diagonals of the diamonds, reversing the arrow in the transformation depicted in Figure~\ref{fig:local-replacement}.

The overall coarsening algorithm performs the following steps:
\begin{enumerate}
\item Test for each vertex $v$ in the initial mesh whether $v$ is coarsenable, and initialize a queue $Q$ of coarsenable vertices.
\item Repeat until $Q$ is empty:
\begin{enumerate}
\item Find and remove a vertex $v$ from $Q$.
\item Perform a coarsening step at $v$.
\item For each neighbor $u$ of $v$ in the coarsened grid, test whether $u$ is now coarsenable, and if it is then add $u$ to $Q$.
\end{enumerate}
\end{enumerate}

\begin{theorem}
The mesh resulting from this adaptation procedure is the coarsest mesh consistent with the new size function, and the time for the procedure is proportional to the sum of the sizes of the old and new meshes.
\end{theorem}

\begin{proof}
The mesh resulting from this adaptation procedure corresponds to a lower set $L$ in the partially ordered set $P_T$, such that $L$ contains only those refinement operations that are either directly necessary to eliminate an oversized quadrilateral, or that are prerequisites of these necessary refinements. Therefore, the mesh resulting from this adaptation process is the unique coarsest mesh that conforms to the local size function, the same mesh that would be generated by applying size refinement to a fixed initial mesh.

The time for adapting a mesh to the changed local size function is proportional to the size of the old mesh (to determine the initial queue for the refinement and unrefinement algorithms) plus the number of refinement and unrefinement steps by which the old and new meshes differ. Each refinement step increases the size of the mesh, so the total spent in the refinement phase is proportional to the number of elements that belong to the new mesh and not the old one. By a similar (but time-reversed) argument, the time spent in the unrefinement phase is proportional to the number of elements that belong to the old mesh and not the new one. Thus, the total time is as stated.
\end{proof}

\subsection{Implementation}

\begin{figure}[t]
\centering\includegraphics[width=3in]{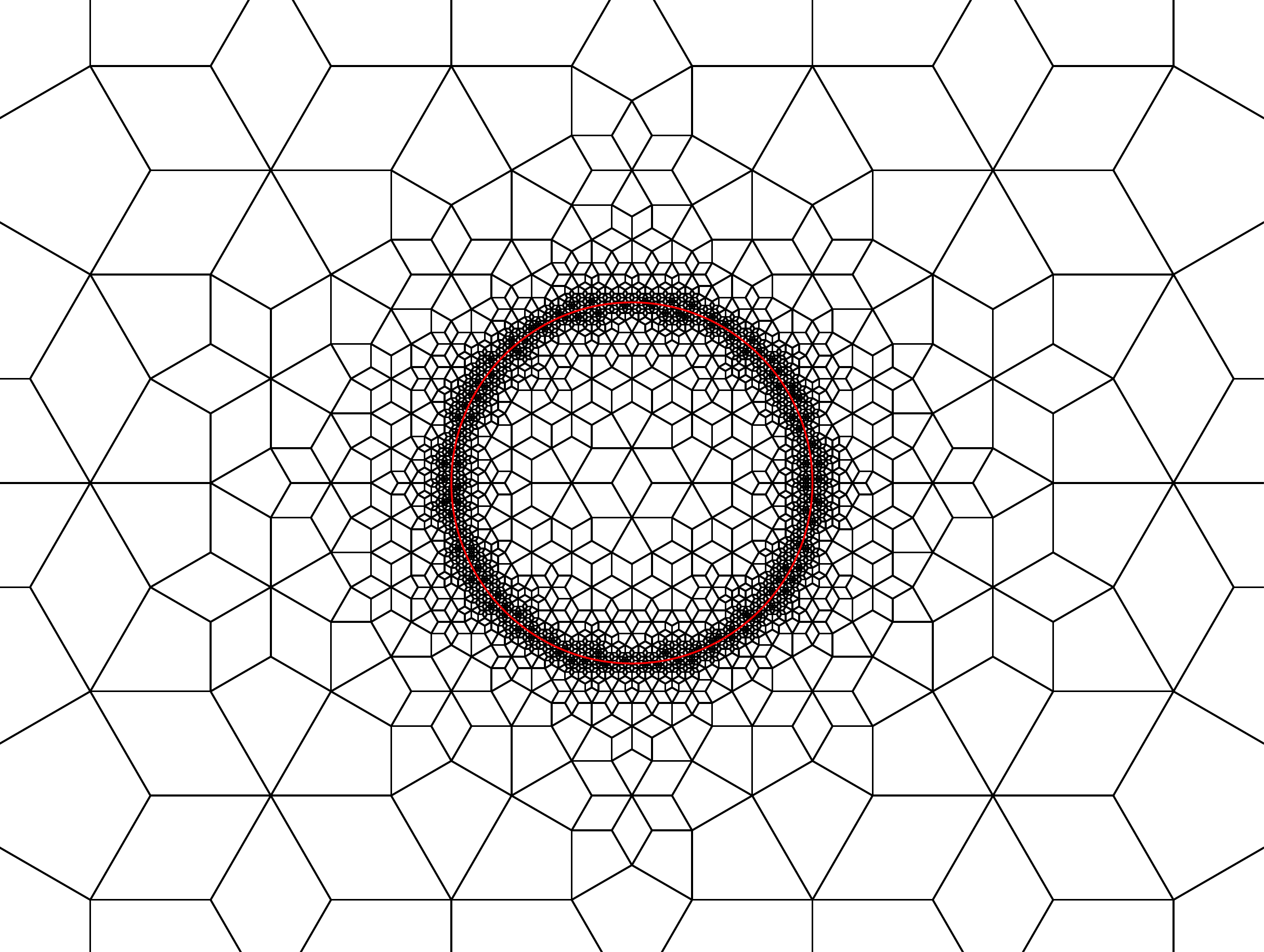}
\caption{The result of applying the refinement algorithm to a size function that measures the distance to a circle.}
\label{fig:circle-limit}
\end{figure}

As a proof of concept, we implemented the refinement algorithm described in Section~\ref{sec:refinement}, using a hardcoded local size function that measures the distance of its argument to a circle. Because of the simple form of this size function, we used a simplified variant of the adaptation algorithm that measures the size function only at vertices of the mesh when deciding whether to call $\refine$, rather than measuring the minimum value of the size function within quadrilaterals of the mesh. Our implementation uses approximately 15 lines of Python code to set up the initial mesh and queue, four lines for the size function, 22 lines for the $\refine$ subroutine, five lines for the refinement algorithm that adapts the mesh to the local size function, and eight lines to output the mesh to the SVG graphics format.  The result, a mesh that is coarse far from the circle and fine near the circle, is depicted in Figure~\ref{fig:circle-limit}.

\section{Properties}
\subsection{Size and length optimality}

Following Ruppert~\cite{Rup-Algs-95}, we may use \emph{local feature size} to prove that the methods for fitting a diamond-kite mesh to a local size function discussed in the previous section produce meshes that (compared to any other mesh with quadrilaterals or triangles of bounded aspect ratio obeying the constraints of the local size function) are within a constant factor of optimal with respect both to their number of elements and to their total edge length, matching known results for quadtree meshes~\cite{BerEppGil-JCSS-94,Epp-DCG-94} and for meshes formed by Delaunay refinement~\cite{Rup-Algs-95}.

We assume that the size function $\sigma(p)$ is defined in such a way as to lead to a finite mesh, and
we define the \emph{local feature size} to be a function $\hat\sigma(p)$ that maps a point $p$ to the number
$$\hat\sigma(p)=\inf\bigl\{\distance(p,q)+\sigma(q)\bigr\}.$$
The point $q$ in the minimization ranges over the rest of the plane, but its minimum will necessarily occur within a disk of radius $\sigma(p)$ centered at $p$, because all other points lead to larger values than the value $\sigma(p)$ achieved at $q=p$. The following two observations are central to our analysis:
\begin{itemize}
\item In all meshes with bounded-aspect-ratio elements, the size of the element containing a given point $p$ is $\Omega(\hat\sigma(p))$. To see this, let $q$ be a point achieving (or approximately achieving) the minimum value in the definition of $\hat\sigma(p)$. The element containing $q$ must have size at most $\hat\sigma(p)$, and the sequence of elements crossed by the line segment from $q$ to $p$ cannot increase in size from that value by more than a constant factor before they reach $p$.
\item In the diamond-kite mesh defined from the size function $\sigma$, the size of the element containing $p$ is $O(\hat\sigma(p))$. This follows from the fact that, if $q$ is any point in the plane, the size of the smallest element of the partial order $P_T$ that is forced by the value of $\sigma(q)$ to be included in the mesh and that corresponds to a local replacement for a hexagon containing $p$ is $O(\distance(p,q)+\sigma(q)$.
\end{itemize}

\begin{theorem}
The diamond-kite mesh for a given size function has a number of elements within a constant factor of the minimum possible for any bounded-aspect-ratio mesh for the same size function.
\end{theorem}

\begin{proof}
Given a bounded-aspect-ratio mesh $M$, let $\alpha(p)$ denote the area of the element containing $p$. Then, for all elements $E$ of $M$, we have the identity
$$\int_E \frac{1}{\alpha(p)} dp=1.$$
Therefore, the number of elements in the mesh can be counted by
$$\int_M \frac{1}{\alpha(p)} dp.$$
However, $1/\alpha(p)$ is lower-bounded by $\Omega((\hat\sigma(p))^{-2})$ for all bounded-aspect-ratio meshes, and upper-bounded by $O((\hat\sigma(p))^{-2})$ for diamond-kite meshes; therefore, the number of elements in the diamond-kite mesh determined by size function $\sigma$ is within a constant factor of optimal.
\end{proof}

\begin{theorem}
The diamond-kite mesh for a given size function has total edge length within a constant factor of the minimum possible for any bounded-aspect-ratio mesh for the same size function.
\end{theorem}

\begin{proof}
Given a bounded-aspect-ratio mesh $M$, let $\pi(p)$ denote the perimeter of the element containing $p$. Then, for all elements $E$ of $M$, we have the identity
$$\int_E \frac{\pi(p)}{\alpha(p)} dp=\mathop{\mathrm{perimeter}}(p).$$
Therefore, the total perimeter of all the elements in the mesh is
$$\int_M \frac{\pi(p)}{\alpha(p)} dp.$$
However, $\pi(p)/\alpha(p)$ is lower-bounded by $\Omega(1/\hat\sigma(p))$ for all bounded-aspect-ratio meshes, and upper-bounded by $O(1/\hat\sigma(p))$ for diamond-kite meshes; therefore, the total perimeter of the elements in the diamond-kite mesh determined by size function $\sigma$ is within a constant factor of optimal.
\end{proof}

\begin{figure}[t]
\centering\includegraphics[height=2.5in]{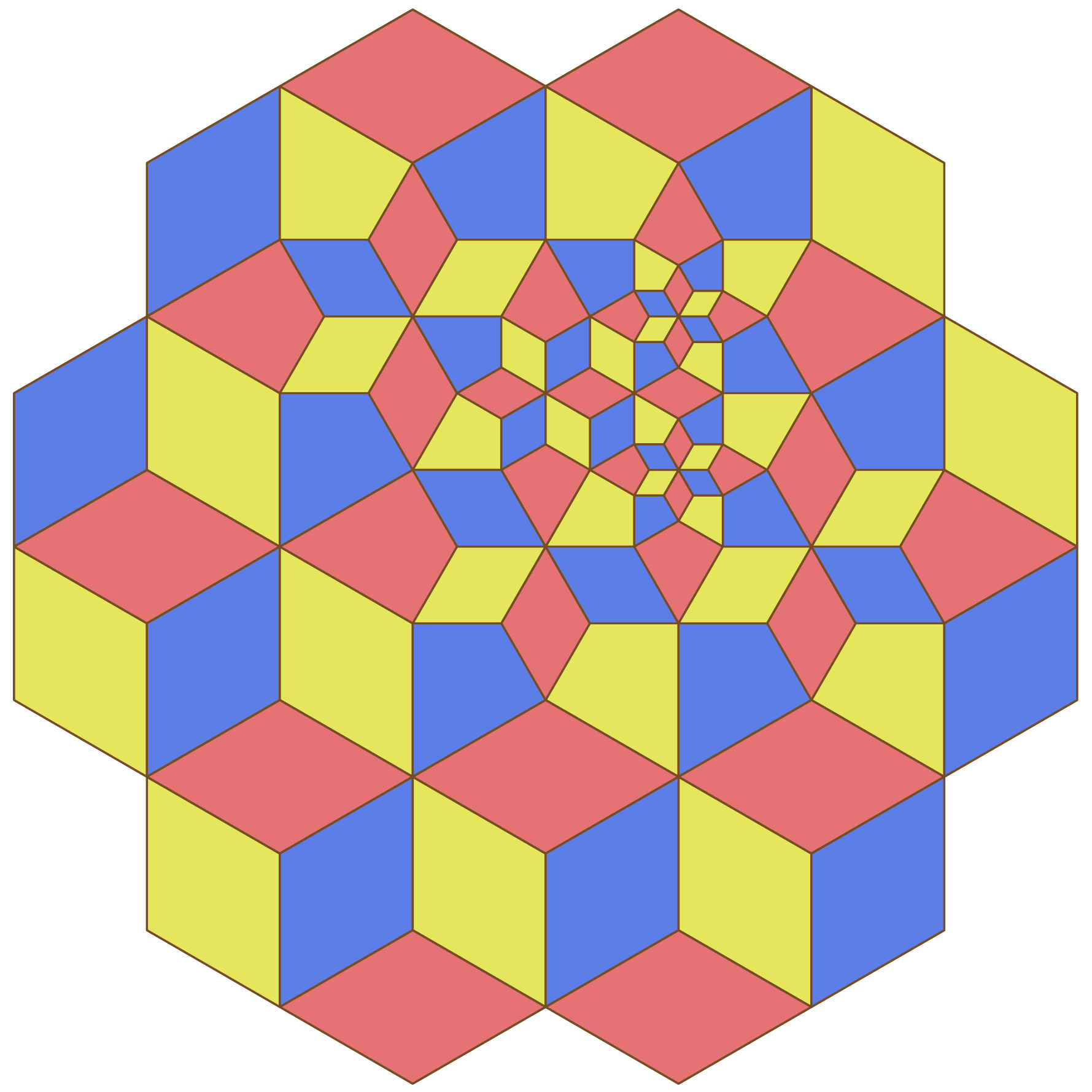}
\caption{A diamond-kite mesh,  3-colored according to the orientation of the diagonals in each quadrilateral.}
\label{fig:3-colored}
\end{figure}

\subsection{Coloring}
Graph colorings of meshes may be used to schedule batches of parallel updates to the values stored at mesh elements, in order to ensure that each two values that are updated in the same batch are independent from each other~\cite{BenFlaKri-AMHCS-92,BerEppHut-Algo-02}.
As with all planar graphs in which every face has an even number of sides, the vertices of a diamond-kite mesh may be colored with two colors, but in this context it is more relevant to color the faces of the mesh so that no two faces that share an edge have the same color. 

\begin{theorem}
The faces of every diamond-kite mesh may be colored with three colors.
\end{theorem}

\begin{proof}
We may find such a coloring by the following simple strategy: define an equivalence relation on the quadrilaterals of the mesh, according to which two quadrilaterals are equivalent when their diagonals are parallel, and assign one color to each equivalence class. There are only three equivalence classes: quadrilaterals in two different equivalence classes will have their diagonals rotated by $30^\circ$ from each other, and after three such rotations we return to the starting equivalence class. No two adjacent quadrilaterals in a diamond-kite mesh may have parallel diagonals, so adjacent quadrilaterals are always assigned distinct colors. Therefore, the result is a proper 3-coloring, as depicted in Figure~\ref{fig:3-colored}.
\end{proof}

In contrast, other kinds of quadrilateral mesh may sometimes require four colors; for example, this is true of the quadrilateral mesh depicted in Figure~\ref{fig:4-chromatic}. In the figure, the red quadrilateral at the upper left shares two neighbors with the yellow quadrilateral in the right corner of the figure; these shared neighbors are also adjacent to each other. If the mesh could be 3-colored, these two neighbors would force the red and yellow quadrilaterals to have the same color as each other. But by a symmetric argument, if the mesh could be 3-colored, the blue quadrilateral in the lower left and the yellow quadrilateral in the right corner would also have to have the same color as each other. Therefore the two leftmost quadrilaterals would have to be colored the same as each other, not possible in a 3-coloring as they share an edge. Therefore, at least four colors are necessary for this mesh, but as shown in the figure four colors are also sufficient. More generally, all quadrilateral meshes require at most four colors, a fact that follows either from the four-color theorem for planar graphs or from Brooks' theorem on coloring regular graphs.

\begin{figure}[t]
\centering\includegraphics[height=1.75in]{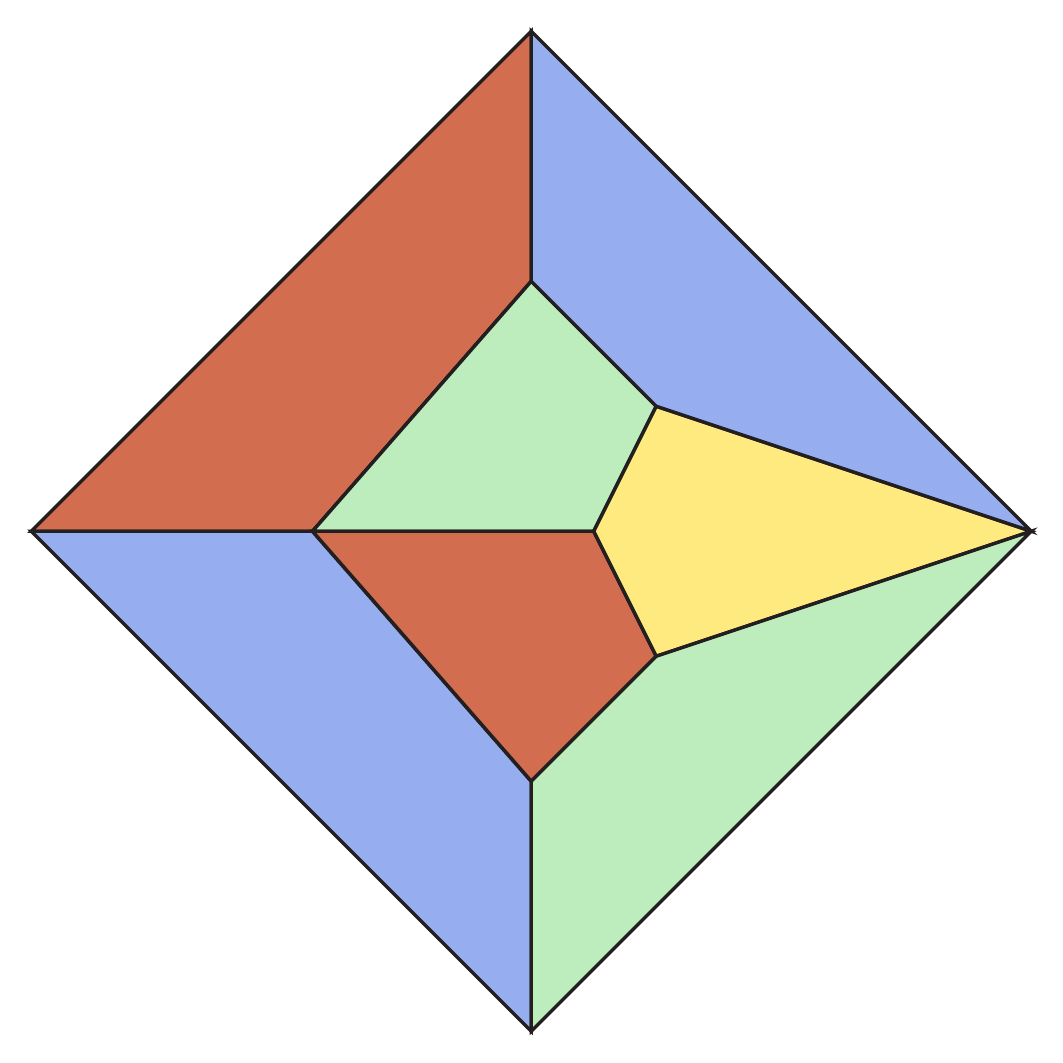}
\caption{A quadrilateral mesh that cannot be colored with fewer than four colors.}
\label{fig:4-chromatic}
\end{figure}

\subsection{Orthogonal circle packing}

A famous and deep theorem of Koebe, Andreev, and Thurston~\cite{And-MSN-70-a,And-MSN-70-b,Koe-BSAWL-36,Thu-MSRI-02} asserts that the vertices of every planar graph may be represented by a \emph{circle packing}, a system of circles with disjoint interiors, such that two vertices are adjacent in the graph if and only if the corresponding two circles are tangent. This representation is not unique without additional constraints (for instance, a 4-cycle has infinitely many distinct representations as a set of four tangent circles) but it can be made unique, up to M\"obius transformations, in one of two different ways:
\begin{itemize}
\item Let $G$ be constrained to be a \emph{maximal} planar graph; that is, every face of $G$, including the outer face, must be a triangle. Then its representation as a circle packing exists and is unique up to M\"obius transformations (Thurston, Corollary 13.6.2). We call this a \emph{maximal circle packing}. An example is shown in Figure~\ref{fig:packings}, left.
\item Alternatively, let $G$ be a \emph{$3$-vertex-connected} planar graph. It has a unique planar embedding; let $G'$ be the dual graph of this embedding. Then it is possible to represent both $G$ and $G'$ by simultaneous circle packings with the property that, for every edge $e$ of $G$ and its corresponding dual edge $e'$, the two circles representing the endpoints of $e$ have the same point of tangency as the two circles representing the endpoints of $e'$ and, moreover, the circles for $e$ cross the circles for $e'$ at right angles at this point. Again, this representation is unique up to M\"obius transformations~\cite{BriSch-SJDM-93}, and we call it an \emph{orthogonal circle packing}. An example is shown in Figure~\ref{fig:packings}, right.
\end{itemize}

\begin{figure}[t]
\centering\includegraphics[width=1.5in]{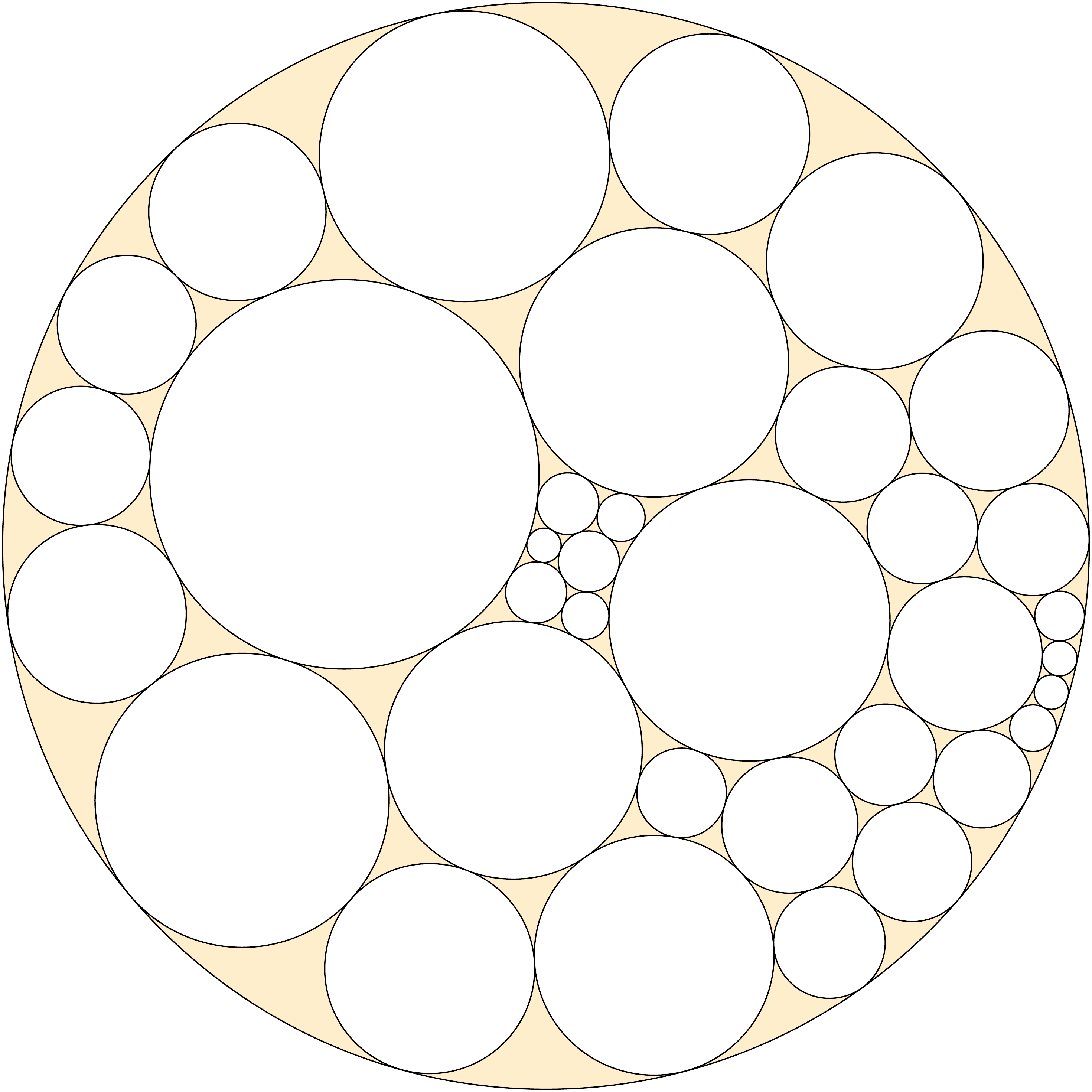}
\qquad\includegraphics[width=1.5in]{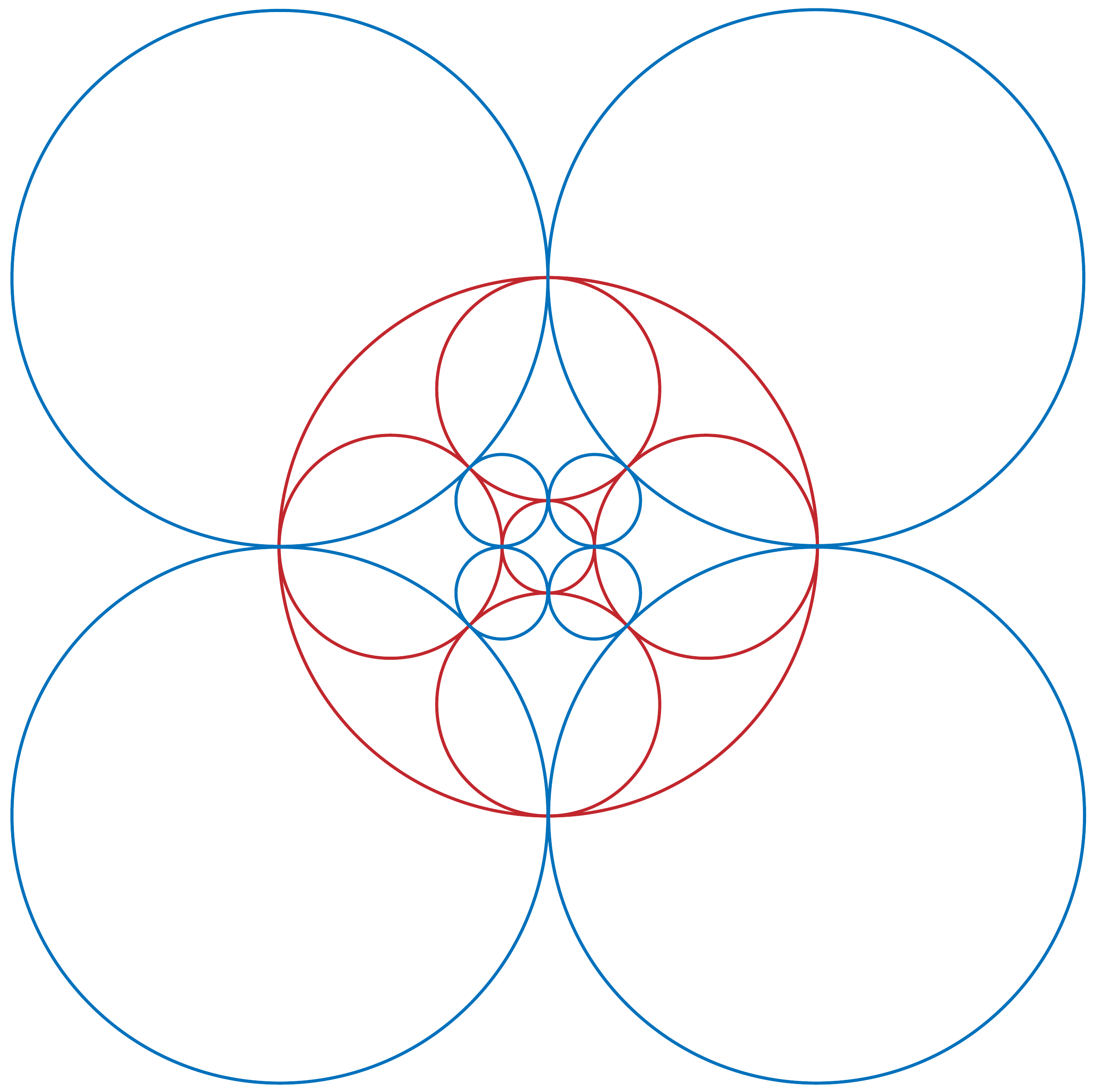}
\caption{A maximal circle packing (left) and an orthogonal circle packing (right).}
\label{fig:packings}
\end{figure}

Circle packings of these types have been applied in the field of graph drawing, to find drawings of planar graphs with right angle crossings~\cite{BriSch-SJDM-93}, high angular resolution~\cite{Epp-GD-12,MalPap-SJDM-94}, and small numbers of distinct slopes~\cite{KesPacPal-GD-10}. They have also been used for mesh partitioning~\cite{EppMilTen-FI-95,MilTenThu-JACM-97,MilTenThu-SJSC-98},
for visualization of brain structures~\cite{HurBowSte-MICCAI-99}, for analyzing the structure of soap bubbles~\cite{Epp-bubbles}, for solving differential equations~\cite{He-TAMS-90}, for constructing Riemann surfaces from combinatorial data~\cite{BowSte-MAMS-04}, and for finding approximations to conformal mappings between different simply connected domains, which can be used as an important step in structured mesh generation~\cite{RodSul-JDG-87,Ste-CMFT-97}.

The two constrained forms of circle packing guaranteed to exist by the circle packing theorem would seem to be also a natural fit for unstructured mesh generation: in a maximal circle packing, the graph of adjacencies between tangent circles (with its vertices placed at the triangle centers) forms an unstructured triangle mesh, and in an orthogonal circle packing, the graph of adjacencies between orthogonal circles forms an unstructured quadrilateral mesh. Additionally, if the degree of a graph is bounded, then the circle packings generated from it are naturally graded in size: adjacent circles have radii whose ratio is bounded, and the triangular or quadrilateral elements derived from the packing have bounded aspect ratio. However, despite their obvious appeal, these types of circle packing have not been used in mesh generation, because the geometry of a circle packing is difficult to control: circle packings are generated from combinatorial data (a graph) rather than from geometric data (the shape of a domain to be meshed) and in general, a small localized change to the graph from which the circle packing is generated can lead to large and non-localized changes to the packing.

\begin{figure}[t]
\centering\includegraphics[width=3in]{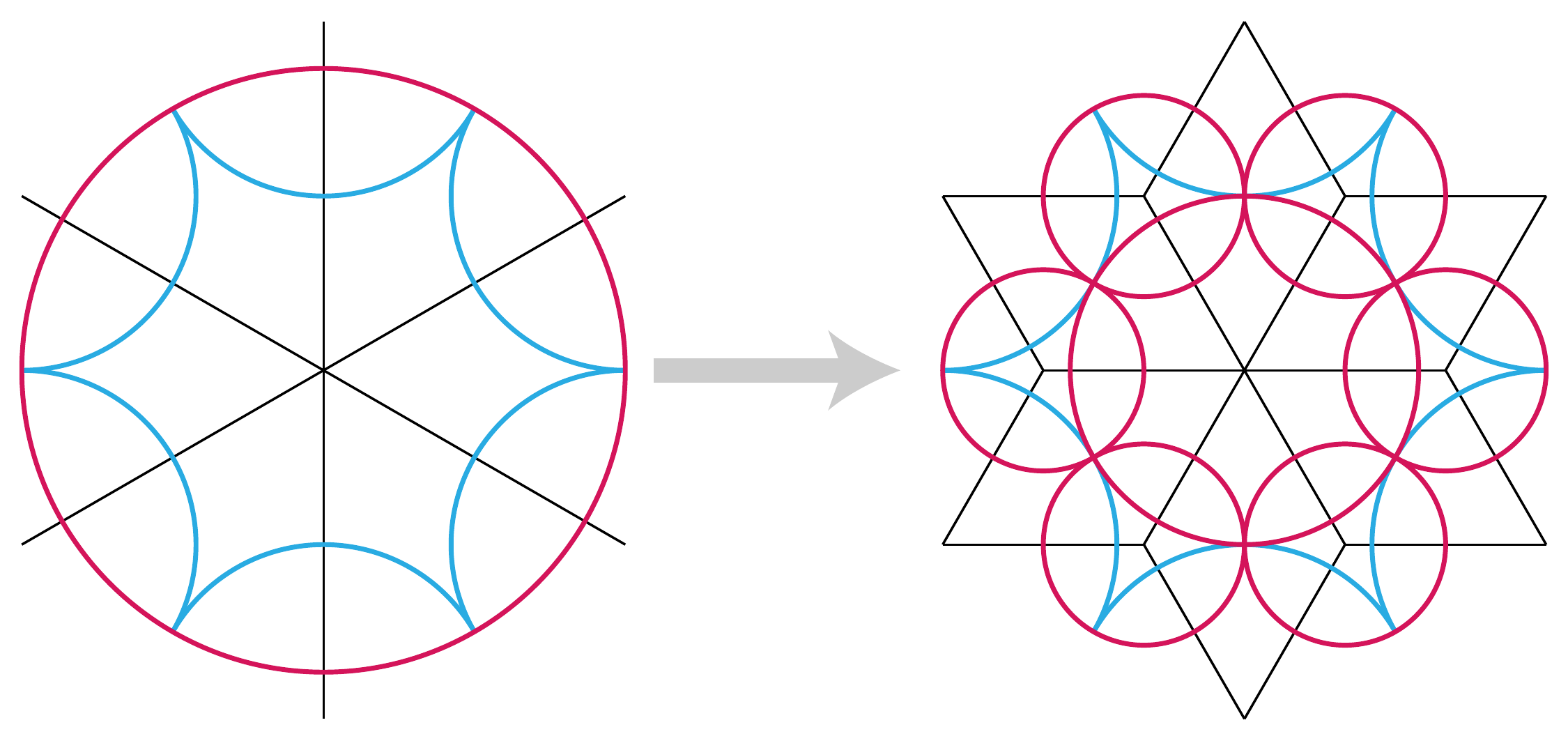}
\caption{Changes to the system of circular arcs within each quadrilateral caused by a single local replacement step.}
\label{fig:replacement-circles}
\end{figure}

Using diamond-kite meshes we show for the first time that it is possible to construct orthogonal circle packings, one of the two types of circle packing guaranteed to exist by the Koebe--Andreev--Thurston theorem, with geometric rather than graph-theoretic control of the position and size of the circles.

\begin{theorem}
The vertices of every diamond-kite mesh form the centers of the circles in an orthogonal circle packing.
\end{theorem}

\begin{proof}
In each quadrilateral of a diamond-kite mesh, place arcs of four circles, centered at the quadrilateral's four vertices and meeting at the center of the quadrilateral. Then, the circular arcs for the quadrilaterals meeting at a vertex will necessarily link up to form a single circle. In the unsubdivided rhombille tiling, this is true because the quadrilaterals sharing a vertex are all rotated images of each other, and it remains true in each of the local replacement steps by which the subdivided tiling is formed. As shown in Figure~\ref{fig:replacement-circles}, the circular arcs surrounding each vertex of the replaced hexagon (shown as green in the figure) retain their previous radius, and the circular arcs surrounding the center vertex and each newly added vertex (shown as violet) meet up to form a circle that lies entirely within the replacement region.

The circles formed in this way meet in tangent pairs at the points within each quadrilateral where the diagonals cross, and (because the rhombs and kites of the mesh are both orthodiagonal) the two pairs of tangent circles meeting at each crossing point are orthogonal to each other. Thus, the result is an orthogonal circle packing.
\end{proof}

In this packing, every two orthogonal circles have radii differing by a factor of exactly $\sqrt 3$, and every two tangent circles have radii differing by a factor of at most three.
Figure~\ref{fig:circle-packing} shows a larger example.

\begin{figure}[t]
\centering\includegraphics[width=3in]{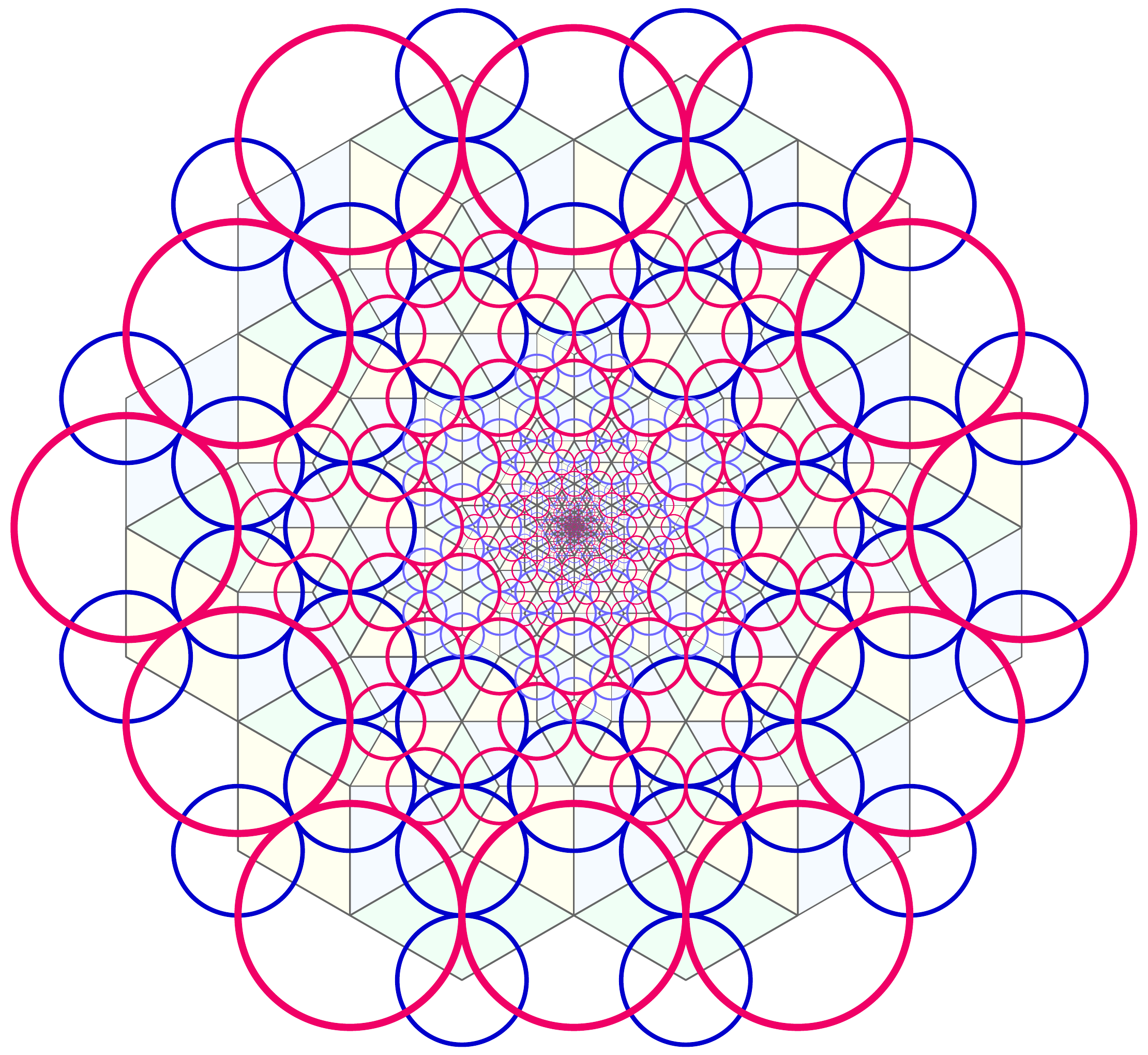}
\caption{A diamond-kite mesh with quadrilaterals of many different scales and its circle packing.}
\label{fig:circle-packing}
\end{figure}

\subsection{Laplacian smoothing}

\emph{Mesh smoothing} is a mesh improvement method in which mesh vertices are moved, without changing the mesh topology, in order to improve the shapes of the mesh elements.
In \emph{Laplacian smoothing}~\cite{Her-JEMD-76}, each mesh vertex is moved to the centroid of its neighbors. It dates back to the work of Tutte~\cite{Tut-PLMS-63}, who showed that for two-dimensional convex domains with boundary vertices fixed in place, this method converges to a unique optimal mesh in which every element is convex. The mesh to which this process converges is called a \emph{Tutte embedding}; it has the property that, if Laplacian smoothing is applied to it, every vertex remains in the same position.
Although other methods including centroidal Voronoi tessellation~\cite{DuGun-AMC-02} and optimization-based mesh smoothing~\cite{AmeBerEpp-Algs-99} have been applied to triangle meshes, and non-linear quasi-Newton smoothing has been shown to be more effective at removing concavities from quadrilateral meshes~\cite{VerTau-IMR-12}, Laplacian smoothing remains a popular choice for quadrilateral mesh smoothing.

As we now argue, diamond-kite meshes already form Tutte embeddings: each vertex lies at the exact centroid of its neighbors, so Laplacian smoothing can make no change to the mesh. In other words, from the point of view of Laplacian smoothing, diamond-kite meshes are already optimally smooth.

\begin{theorem}
In a diamond-kite mesh, every vertex lies at the centroid of its neighbors.
\end{theorem}

\begin{proof}
It is straightforward to verify that in the rhombille tiling used as the starting point for a diamond-kite mesh, each vertex lies at its neighbors' centroid: each vertex $v$ has either three or six equally distant neighbors, and the edges to those neighbors are equally spaced around $v$, so the vectors representing the differences in position between $v$ and its neighbors sum to zero. When we perform the replacement step depicted in Figure~\ref{fig:local-replacement}, this property is preserved for every vertex $v$, as can be seen by a simple case analysis:
\begin{itemize}
\item If $v$ is the central vertex of a replacement step, it is surrounded after the replacement by six equally spaced and equal length edges, so just as in the initial mesh the vectors representing the differences in position between $v$ and its neighbors sum to zero.
\item If $v$ is one of the six newly added neighbors of the central vertex, it is surrounded after the replacement by three equally spaced and equal length edges, and the same argument applies.
\item If $v$ is one of the six vertices of the hexagon within which the replacement happens, its neighborhood is changed by the removal of an edge $e_1$ and the addition of two edges $e_2$ and $e_3$, with length $1/\sqrt 3$ times that of $e_1$ and forming angles of $\pm 30^\circ$ with $e_1$.  The two new edges $e_2$ and $e_3$ form two consecutive sides of a rhombus, and $e_1$ is the diagonal of the rhombus that lies between them. As with any rhombus (or more generally any parallelogram) the vectors from $v$ to the endpoints of $e_2$ and $e_3$ sum to give the vector from $v$ to the endpoint of $e_1$, so the refinement step does not change the sum of the vectors from $v$ to its neighbors, which remains zero. 
\item In all other cases, the neighbors of $v$ are not changed by the refinement step.
\end{itemize}
Because the initial coarse mesh used as the starting point for producing a diamond-kite mesh is a Tutte embedding, and because every refinement step preserves this property, it follows that every diamond-kite mesh is a Tutte embedding.
\end{proof}

\section{Well-centered meshes}

\begin{figure}[t]
\centering\includegraphics[height=2.5in]{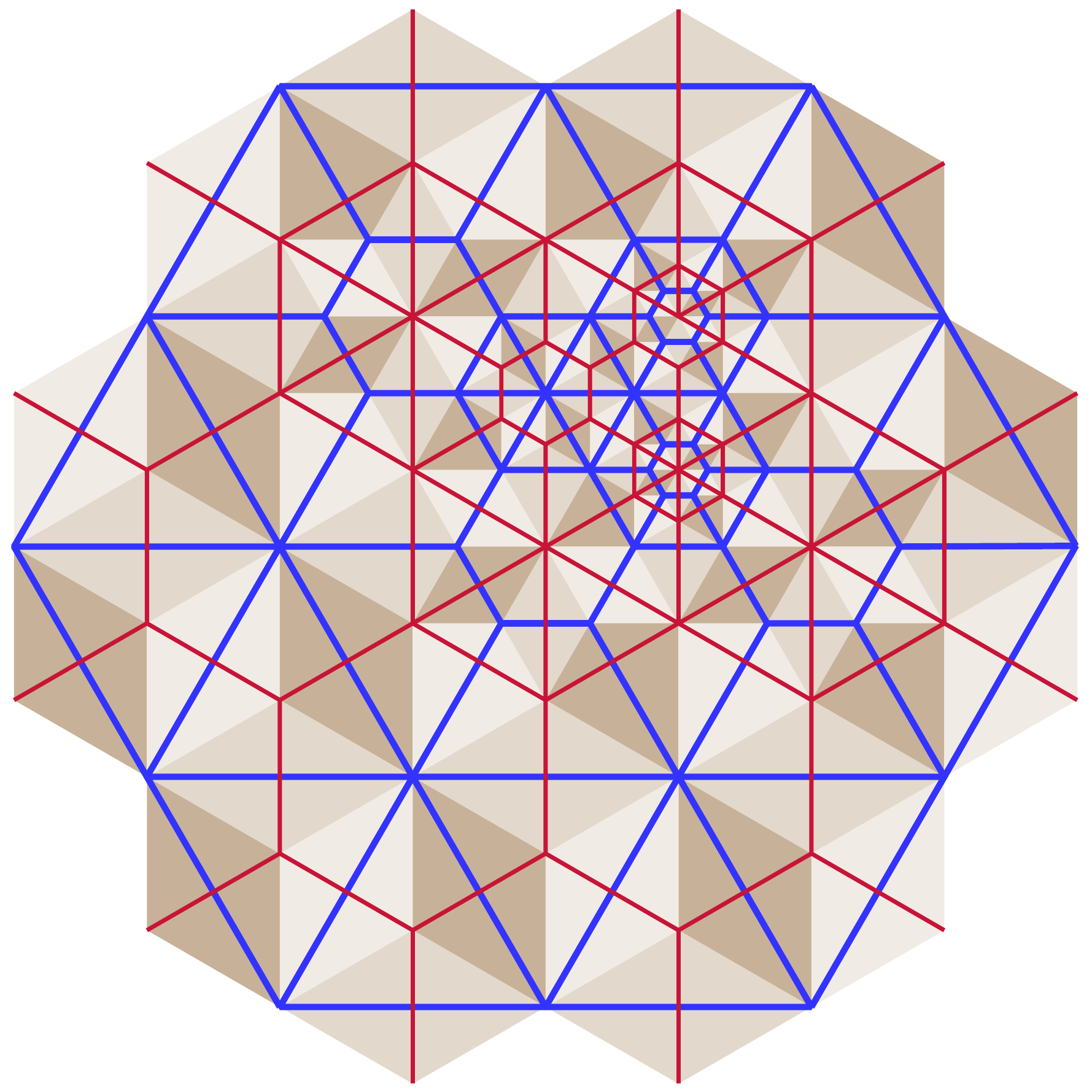}
\caption{Two dual well-centered meshes (the red and blue edges) generated from the diagonals of a diamond-kite mesh (shown by the light brown shading).}
\label{fig:well-centered}
\end{figure}

Vanderzee, Hirani, and Guoy~\cite{VanHirGuo-SJSC-09} have studied \emph{well-centered} triangle meshes, meshes in which the circumcenter of every triangle is interior to the triangle, and their higher-dimensional generalizations. In the plane, such a mesh must be the Delaunay triangulation of its vertices, and its dual is their Voronoi diagram. The Voronoi vertices (dual to the Delaunay triangles) each lie within the interiors of their triangles by definition, and (as with every Voronoi diagram) the triangle vertices lie within the interiors of their Voronoi cells. Additionally, every edge of the Delaunay triangulation crosses the corresponding dual edge of the Voronoi diagram, and the crossing is at right angles. As Vanderzee et al. argue, this orthogonal crossing property of the primal and dual meshes is important for certain numerical methods, including the covolume method, the discrete exterior calculus, and space-time meshing.

Diamond-kite meshes can be used to generate a similar pair of dual meshes, with the same well-centered properties: each primal vertex is interior to the corresponding dual face, each dual vertex is interior to the corresponding primal face,  each primal-dual pair of edges crosses orthogonally, and there are no other crossings. To do so, observe that (because it is a planar graph with faces that have an even number of sides) the diamond-kite vertices and edges form a bipartite graph. Each diagonal of one of the quadrilaterals in the mesh connects two vertices of the same color. Therefore, the set of all diagonals can be partitioned into two disjoint meshes, one of each color (Figure~\ref{fig:well-centered}). Each vertex of one color is surrounded by a cycle of diagonals of the other color (the diagonals that it is not an endpoint of among the diamonds and kites that surround it), so in these two diagonal meshes, each vertex of one mesh lies interior to a face of the other mesh. Each edge of one diagonal mesh crosses orthogonally a single dual edge of the other diagonal mesh, the one from the same diamond or kite. Therefore, these two diagonal meshes are well-centered in the sense of Vanderzee et al.

Like the underlying diamond-kite mesh, these two well-centered meshes have elements whose size varies proportionally to a given local size function. Unlike the well-centered triangle meshes (in which the Delaunay triangulation and the Voronoi diagram are meshes of two different types, and the shapes of the triangles and dual Voronoi cells may vary continuously) these meshes both have elements of the same four shapes: equilateral triangles, regular hexagons, isosceles trapezoids (half-hexagons), and the pentagons formed by gluing together an equilateral triangle and a hexagon.

\section{Conclusions}

We have defined the family of diamond-kite meshes based on a simple local replacement step starting with the rhombille tiling. In these meshes, the most acute angle is $60^\circ$, the most obtuse angle is $120^\circ$, and all elements have bounded aspect ratio. The element size can be controlled by a local size function, and the number of elements and total edge length of the elements is within a constant factor of optimal for the given size function. Replacement operations may be performed adaptively to handle time-dependent size functions. Unlike adaptive quadtree meshes, this system provides a quadrilateral mesh directly without any need for additional subdivision.

The vertices of our new meshes form the centers of an orthogonal circle packing of the type guaranteed to exist by the Koebe--Andreev--Thurston circle packing theorem, showing for the first time that it is possible to incorporate this type of circle packing into a meshing algorithm.

Much remains to be studied in this area. On the mathematical side, we still do not know whether it is possible to define an analogous local replacement scheme that would allow the generation of maximal circle packings (in which every gap between three circles has exactly three sides) with similar properties to those of the diamond-kite mesh, and in particular with the ability to control the size of the circles in one part of the packing without changing the geometry of the circles in distant parts of the packing. On the more practical side, it would be of interest to develop the diamond-kite method into a practical mesh generation system and to compare empirically the quality of meshes generated in this way with those from other comparable systems such as quadtrees and Delaunay refinement. One complication here is that diamond-kite meshes have edges with a small fixed set of orientations, making it difficult for them to conform to the boundary of a domain. We leave such developments to future research.

\raggedright
\ifJournal
\bibliographystyle{spmpsci}
\else
\bibliographystyle{abuser}
\fi
\bibliography{diamond-kite}

\end{document}